\documentclass[english]{paper}
\usepackage[T1]{fontenc}
\usepackage[latin9]{inputenc}
\usepackage{textcomp}
\usepackage{amsthm}
\usepackage{amsmath}
\usepackage{amssymb}

\makeatletter
\newcommand{\lyxaddress}[1]{
\par {\raggedright #1
\vspace{1.4em}
\noindent\par}
}
\theoremstyle{plain}
\newtheorem{thm}{\protect\theoremname}

\makeatother

\usepackage{babel}
\providecommand{\theoremname}{Theorem}

\begin{document}

\title{Quantum Financial Economics of Games of Strategy and Financial Decisions}

\author{Carlos Pedro Gonçalves}

\institution{{\small Instituto Superior de Ciências Sociais e Políticas (ISCSP)
- Technical University of Lisbon}}

\maketitle

\lyxaddress{cgoncalves@iscsp.utl.pt}
\begin{abstract}
A quantum financial approach to finite games of strategy is addressed,
with an extension of Nash's theorem to the quantum financial setting,
allowing for an entanglement of games of strategy with two-period
financial allocation problems that are expressed in terms of: the
consumption plans' optimization problem in pure exchange economies
and the finite-state securities market optimization problem, thus
addressing, within the financial setting, the interplay between companies'
business games and financial agents' behavior.

A complete set of quantum Arrow-Debreu prices, resulting from the
game of strategy's quantum Nash equilibrium, is shown to hold, even
in the absence of securities' market completeness, such that Pareto
optimal results are obtained without having to assume the completeness
condition that the rank of the securities' payoff matrix is equal
to the number of alternative lottery states.\end{abstract}

\section{Introduction}

Underlying economic theory, the notion of game constitutes a fundamental
conceptual unit with operativity in approaching the systemic behavior
and dynamics of economic systems, an argument that was established
by von Neumann and Morgenstern in \cite{key-12}. Each (formal) game
structure attempts to formalize complex decisional problems that demand
a calculatory adaptiveness on the part of the players, a main point
which is present both in standard game theory \cite{key-10,key-12}
as well as in the quantum expansions \cite{key-9,key-11}.

Financial economics has approached, foundationally, the decisional
contexts for the problem of allocation of resources, from an exposure
to games comprised of pure lotteries, that is, a game against nature
where nature {}``chooses'' each state with a certain probability
\cite{key-1,key-5,key-8}.

However, this constitutes an oversimplification of financial systems
and their interplay with economic systems. Since companies play games
of strategy and, ultimately, a securities market for those companies'
shares is about an exposure to such strategic contexts.

Thus, to better understand the consequences of such an interplay between
business games and financial systems' resources and wealth allocation
problems \cite{key-1,key-5,key-8}, we change, in the present work,
the fundamental building block of the \emph{pure lottery game}, to
make the lottery entangled with a quantum game of strategy, and address
how the financial economics of quantum games of strategy may impact
on the structure and conclusions of the two basic financial decision
problems: \textbf{the management and optimization of consumption plans
in pure exchange economies} and the \textbf{securities' market optimization
problem} \cite{key-1,key-5}, for securities that lead to exposure,
on the part of the financial agents, to each player's position in
the quantum game.

In \textbf{section 2.}, we address the quantum financial economics
of finite games of strategy and provide for an illustrative example
of a game between two companies, showing how one may interpret both
the disentangled quantum mixed strategies equilibrium as well as the
quantum entangled solution.

In \textbf{section 3.}, we draw upon the formal work of \textbf{section
2.} and introduce a quantum game of strategy-dependent lottery, leading
to a quantum entanglement between a pure exchange economy for consumption
claims on the lottery results and the underlying game of strategy.
In this way, it its shown how one can operationalize the traditional
two-period exchange economy results of standard financial theory within
a setting in which the lottery is entangled with a strategic choice
problem.

From these results, the securities market portfolio problem is addressed,
being shown that, when each security offers an exposure to a game
position, Pareto optimality for the financial agents' consumption
problem is guaranteed, even in incomplete financial markets, as long
as the financial prices reflect the quantum game present value of
that exposure.

In \textbf{section 4.}, we conclude with some final remarks regarding
quantum game theory and quantum financial economics.

\section{Quantum Financial Economics of Finite Games of Strategy}

\subsection{The Formalism and Equilibrium Problem}

Finite games of strategy, within the framework of noncooperative quantum
game theory \cite{key-10}, can be approached from finite chain categories,
where, by finite chain category, it is understood a category $\mathcal{C}(n;N)$
that is generated by $n$ objects and $N$ morphic chains, called
primitive chains, linking the objects in a specific order, such that
there is a single labelling. $\mathcal{C}(n;N)$ is, thus, generated
by $N$ primitive chains of the form:

\begin{equation}
x_{0}\overset{f_{1}}{\longrightarrow}x_{1}\overset{f_{2}}{\longrightarrow}x_{1}...x_{n-1}\overset{f_{n}}{\longrightarrow}x_{n}
\end{equation}
A finite chain category is interpreted as a finite game category as
follows: to each morphism in a chain $x_{i-1}\overset{f_{i}}{\longrightarrow}x_{i}$
there corresponds a strategy played by a player that occupies the
position $i$, in this way, a chain corresponds to a sequence of strategic
choices available to the players.

A quantum formal theory, for a finite game category $\mathcal{C}(n;N)$,
is defined as a formal structure such that each morphic fundament
$f_{i}$ of the morphic relation $x_{i-1}\overset{f_{i}}{\longrightarrow}x_{i}$
is a tuple of the form:
\begin{equation}
f_{i}:=\left(\mathcal{H}_{i},\mathcal{P}_{i},\hat{P}_{f_{i}}\right)
\end{equation}
where $\mathcal{H}_{i}$ is the $i$-th player's Hilbert space, $\mathcal{P}_{i}$
is a complete set of projectors onto a basis that spans the Hilbert
space, and $\hat{P}_{f_{i}}\in\mathcal{P}_{i}$. This structure is
interpreted as follows: from the strategic Hilbert space $\mathcal{H}_{i}$,
given the pure strategies' projectors $\mathcal{P}_{i}$, the player
chooses to play $\hat{P}_{f_{i}}$.

From the morphic fundament definition (2), an assumption has to be
made on composition in the finite category, we assume the following
tensor product composition operation:
\begin{equation}
f_{j}\circ f_{i}=f_{ji}
\end{equation}
\begin{equation}
f_{ji}=\left(\mathcal{H}_{ji}=\mathcal{H}_{j}\otimes\mathcal{H}_{i},\mathcal{P}_{ji}=\mathcal{P}_{j}\otimes\mathcal{P}_{i},\hat{P}_{f_{ji}}=\hat{P}_{f_{j}}\otimes\hat{P}_{f_{i}}\right)
\end{equation}
From this definition of composition, a morphism for a \emph{game choice
path} can be introduced as:
\begin{equation}
x_{0}\overset{f_{n...21}}{\longrightarrow}x_{n}
\end{equation}
\begin{equation}
f_{n...21}=\left(\mathcal{H}_{G}=\bigotimes_{i=n}^{1}\mathcal{H}_{i},\mathcal{P}_{G}=\bigotimes_{i=n}^{1}\mathcal{P}_{i},\hat{P}_{f_{n...21}}=\bigotimes_{i=n}^{1}\hat{P}_{f_{i}}\right)
\end{equation}
in this way, the choices along the chain of players are completely
{}``encoded'' in the tensor product projectors $\hat{P}_{f_{n...21}}$.
Given the above definitions, there are $N=\prod_{i=1}^{n}\dim\left(\mathcal{H}_{i}\right)$
such morphisms, a number that coincides with the number of primitive
chains in the category $\mathcal{C}(n;N)$.

Each projector can be addressed as a strategic marker of a game path,
and leads to the matrix form of an Arrow-Debreu security \cite{key-1,key-5},
therefore, we call it \emph{game Arrow-Debreu projector}.

While, in traditional financial economics, the Arrow-Debreu securities
pay one unit of numeraire per state of nature, in the present game
setting, they pay one unit of payoff per game path at the beginning
of the game, however this analogy may be taken it must be addressed
with some care, since these are not securities, rather, they represent,
projectively, strategic choice chains in the game, so that the price
of a projector $\hat{P}_{f_{n...21}}$ (the Arrow-Debreu price) is
the price of a strategic choice and, therefore, the result of the
strategic evaluation of the game by the different players.

Now, let $\left|\Psi\right\rangle $ be a ket vector in the game's
Hilbert space $\mathcal{H}_{G}$, such that:
\begin{equation}
\left|\Psi\right\rangle =\sum_{f_{n...21}}\psi\left(f_{n...21}\right)\left|f_{n...21}\right\rangle 
\end{equation}
where $\psi\left(f_{n...21}\right)$ is the Arrow-Debreu price amplitude,
with the condition: 
\begin{equation}
\sum_{f_{n...21}}\left|\psi\left(f_{n...21}\right)\right|^{2}=D
\end{equation}
for $D>0$, then, the $\left|\psi\left(f_{n...21}\right)\right|^{2}$
correspond to the Arrow-Debreu prices for the game path $f_{n...21}$
and $D$ is the discount factor in riskless borrowing, defining an
economic scale for temporal connections between one unit of payoff
now and one unit of payoff at the end of the game, such that one unit
of payoff now can be {}``capitalized'' to the end of the game (when
the decision takes place) through a multiplication by $\frac{1}{D}$,
while one unit of payoff at the end of the game can be discounted
to the beginning of the game through multiplication by $D$.

In this case, the unit operator $\hat{1}=\sum_{f_{n...21}}\hat{P}\left(f_{n...21}\right)$
has a similar profile as that of a bond in standard financial economics,
with $\left\langle \Psi\left|\hat{1}\right|\Psi\right\rangle =D$,
on the other hand, the general payoff system, for each player, can
be addressed from an operator expansion:
\begin{equation}
\hat{\pi}_{i}=\sum_{f_{n...21}}\pi_{i}\left(f_{n...21}\right)\hat{P}_{f_{n...21}}
\end{equation}
where each weight $\pi_{i}\left(f_{n...21}\right)$ corresponds to
quantities associated with each \emph{Arrow-Debreu projector} that
can be interpreted as similar to the quantities of each Arrow-Debreu
security for a general asset. Multiplying each weight by the corresponding
Arrow-Debreu price, one obtains the payoff value for each alternative
such that the total payoff for the player at the end of the game is
given by:
\begin{equation}
\left\langle \Psi\left|\hat{\pi}_{i}\right|\Psi\right\rangle =\sum_{f_{n...21}}\pi_{i}\left(f_{n...21}\right)\frac{\left|\psi\left(f_{n...21}\right)\right|^{2}}{D}
\end{equation}

We can discount the total payoff in (10) to the beginning of the game
using the discount factor $D$, leading to the present value payoff
for the player:
\begin{equation}
PV_{i}=D\left\langle \Psi\left|\hat{\pi}_{i}\right|\Psi\right\rangle =D\sum_{f_{n...21}}\pi_{i}\left(f_{n...21}\right)\frac{\left|\psi\left(f_{n...21}\right)\right|^{2}}{D}
\end{equation}
In the above equation, the $\pi_{i}\left(f_{n...21}\right)$ represent
quantities, while the ratio $\frac{\left|\psi\left(f_{n...21}\right)\right|^{2}}{D}$
represents the future value (at the decision moment) of the quantum
Arrow-Debreu prices (capitalized quantum Arrow-Debreu prices). Introducing
the ket $\left|Q\right\rangle \in\mathcal{H}_{G}$, such that:
\begin{equation}
\left|Q\right\rangle =\frac{1}{\sqrt{D}}\left|\Psi\right\rangle 
\end{equation}
then, $\left|Q\right\rangle $ is a normalized ket for which the price
amplitudes are expressed in terms of their future value. Replacing
in (11), we have:
\begin{equation}
PV_{i}=D\left\langle Q\left|\hat{\pi}_{i}\right|Q\right\rangle 
\end{equation}

In the quantum game setting, the capitalized Arrow-Debreu price amplitudes$\left\langle f_{n...21}|Q\right\rangle $
become quantum strategic configurations, resulting from the strategic
cognition of the players with respect to the game. Given $\left|Q\right\rangle $,
each player's strategic valuation of each pure strategy can be obtained
by introducing the projector chains:
\begin{equation}
\hat{C}_{f_{i}}=\sum_{f_{n...i+1},f_{i-1...1}}\hat{P}_{f_{n...i+1}}\otimes\hat{P}_{f_{i}}\otimes\hat{P}_{f_{i-1...1}}
\end{equation}
with $\sum_{f_{i}}\hat{C}_{f_{i}}=\hat{1}$. For each alternative
choice of the player $i$, the chain sums over all of the other choice
paths for the rest of the players, such chains are called coarse-grained
chains in the decoherent histories approach to quantum mechanics \cite{key-2,key-4,key-6},
following this approach, one may introduce the pricing functional
from the expression for the decoherence functional \cite{key-2,key-4,key-6}:
\begin{equation}
\mathcal{D}\left(f_{i},g_{i}:\left|Q\right\rangle \right)=\left\langle Q\left|\hat{C}_{f_{i}}^{\dagger}C_{g_{i}}\right|Q\right\rangle 
\end{equation}
we, then, have, for each player:
\begin{equation}
\mathcal{D}\left(f_{i},g_{i}:\left|Q\right\rangle \right)=0,\:\forall f_{i}\neq g_{i}
\end{equation}
this is the usual quantum mechanics' condition for an aditivity of
measure (also known as decoherence condition), which, in the present
case, means that the capitalized prices for two different alternative
choices of player $i$ are additive. Then, we can work with the pricing
functional $\mathcal{D}\left(f_{i},f_{i}:\left|Q\right\rangle \right)$
as giving, for each player an Arrow-Debreu capitalized price associated
with the pure strategy, expressed by $f_{i}$. Given that the condition
(16) is satisfied, each player's quantum Arrow-Debreu pricing matrix,
defined analogously to the decoherence matrix from the decoherent
histories approach, is a diagonal matrix and can be expanded as a
linear combination of the projectors for each player's pure strategies
as follows:
\begin{equation}
\mathbf{D}_{i}\left(\left|Q\right\rangle \right)=\sum_{f_{i}}\mathcal{D}\left(f_{i},f_{i}:\left|Q\right\rangle \right)\hat{P}_{f_{i}}
\end{equation}
which has the mathematical expression of a mixed strategy. In particular,
$\mathbf{D}_{i}\left(\left|Q\right\rangle \right)$ can be regarded
as points in a simplex whose vertices are the pure strategies' projectors.
Gathering all of the pricing matrices $\mathbf{D}_{i}\left(\left|Q\right\rangle \right)$
we obtain the \emph{n}-tuple $\mathfrak{D}\left(\left|Q\right\rangle \right)=\left(\mathbf{D}_{1}\left(\left|Q\right\rangle \right),\mathbf{D}_{2}\left(\left|Q\right\rangle \right),...,\mathbf{D}_{n}\left(\left|Q\right\rangle \right)\right)$.
Introducing $pv_{i}\left[\mathfrak{D}\left(\left|Q\right\rangle \right)\right]:=PV_{i}=\left\langle Q\left|\hat{\pi}_{i}\right|Q\right\rangle $,
Nash's equilibrium definition \cite{key-10}, follows, within the
quantum financial setting, as:
\begin{equation}
pv_{i}\left[\mathfrak{D}\left(\left|Q\right\rangle \right)\right]=\max_{\textrm{all }\left|Q'\right\rangle }\left\{ pv_{i}\left[\mathfrak{D}\left(\left|Q'\right\rangle ;\mathbf{D}_{i}\left(\hat{U}\left|Q'\right\rangle \right)\right)\right]\right\} 
\end{equation}
where by $\mathfrak{D}\left(\left|Q'\right\rangle ;\mathbf{D}_{i}\left(\hat{U}\left|Q'\right\rangle \right)\right)$
it is understood that, given the feasible quantum price strategy $\left|Q'\right\rangle $,
the \emph{i}-th player performs a unitary rotation $\hat{U}$ leading
to a substitution of $\mathbf{D}_{i}\left(\left|Q'\right\rangle \right)$
by $\mathbf{D}_{i}\left(\hat{U}\left|Q'\right\rangle \right)$ leaving
all of the rest of the elements in the \emph{n}-tuple unchanged. Thus,
each player chooses from all of the possible quantum computations,
the one that maximizes the present value payoff function with all
the other strategies held fixed, which is in agreement with Nash \cite{key-10}.

The equilibrium price ket $\left|Q\right\rangle $ is such that each
player's present value payoff is optimal, given all of the other players'
strategies. The following theorem can now be proven:
\begin{thm}
\textbf{\emph{(Quantum Equilibrium Theorem) - Every quantum finite
game has a ket of capitalized Arrow-Debreu price amplitudes, unique
up to a capitalized Arrow-Debreu price-conserving unitary transformation,
that is an equilibrium solution for the game.}}\end{thm}
\begin{proof}
In the present proof, we assume all of the above definitions and game
formalism context. Let, then, $\left|Q\right\rangle \in\mathcal{H}_{G}$
be a ket of capitalized Arrow-Debreu price amplitudes and $\mathfrak{D}\left(\left|Q\right\rangle \right)$
be the corresponding \emph{n}-tuple of pricing matrices, define also
$\left|Q:f_{i}\right\rangle \in\mathcal{H}_{G}$ to be such that $\mathbf{D}_{j}\left(\left|Q:f_{i}\right\rangle \right)=\mathbf{D}_{j}\left(\left|Q\right\rangle \right)$
for every $j\neq i$ and $\mathbf{D}_{i}\left(\left|Q:f_{i}\right\rangle \right)=\hat{P}_{f_{i}}$.
Following Nash's proof in \cite{key-10}, we introduce the set of
continuous functions of $\mathfrak{D}\left(\left|Q\right\rangle \right)$
defined by:
\begin{equation}
\varphi_{if_{i}}\left[\mathfrak{D}\left(\left|Q\right\rangle \right)\right]:=\max\left(0,pv_{i}\left[\mathfrak{D}\left(\left|Q:f_{i}\right\rangle \right)\right]-pv_{i}\left[\mathfrak{D}\left(\left|Q\right\rangle \right)\right]\right)
\end{equation}
Let, now, $\hat{U}$ be a unitary transformation such that $\hat{U}\left|Q\right\rangle =\left|X\right\rangle $
and, for each, $\mathbf{D}_{i}\left(\left|Q\right\rangle \right)\in\mathfrak{D}\left(\left|Q\right\rangle \right)$
it holds that:
\begin{equation}
\mathbf{D}_{i}\left(\left|X\right\rangle \right)=\frac{1}{1+\sum_{f_{i}}\varphi_{if_{i}}\left[\mathfrak{D}\left(\left|Q\right\rangle \right)\right]}\left(\mathbf{D}_{i}\left(\left|Q\right\rangle \right)+\sum_{f_{i}}\varphi_{if_{i}}\left[\mathfrak{D}\left(\left|Q\right\rangle \right)\right]\hat{P}_{f_{i}}\right)
\end{equation}
then, the unitary transformation leads to a pricing matrix transformation
that coincides with the transformation addressed by Nash in \cite{key-10},
such that its fixed points are the equilibrium points. Given the geometric
coincidence between Nash's formulation and the structure of the space
of pricing matrices for the quantum game, it follows that Nash's theorem
applies and every game has an \emph{n}-tuple of equilibrium pricing
matrices.

For any $\mathfrak{D}\left(\left|Q\right\rangle \right)$ which is
an equilibrium, there is a family of kets defined by:
\begin{equation}
\mathcal{E}\left(\mathfrak{D}\left(\left|Q\right\rangle \right)\right)=\left\{ \left|X\right\rangle \in\mathcal{H}_{G}:\mathfrak{D}\left(\left|X\right\rangle \right)=\mathfrak{D}\left(\left|Q\right\rangle \right)\right\} 
\end{equation}
all of the kets in $\mathcal{E}\left(\mathfrak{D}\left(\left|Q\right\rangle \right)\right)$
are related to each other by equilibrium preserving unitary transformations,
which necessarily conserve the capitalized equilibrium quantum Arrow-Debreu
prices.
\end{proof}
In the above theorem, and proof, the probabilistic interpretation
is not invoked for the capitalized prices $\frac{\left|\psi\left(f_{n...21}\right)\right|^{2}}{D}$,
which have a mathematical structure, with respect to the payoff operators,
similar to probability measures. Within an evolutionary finance framework,
one can, however, offer a probabilistic interpretation by working
with a notion of \emph{fitness} as a selectibility of a strategic
solution in an adaptation problem \cite{key-3}, such that the probability
of a strategy being selected is proportional to its adaptive value.

In the present case, the capitalized quantum Arrow-Debreu prices constitute
a strategic valuation measure of a game path at the end of the decision
frame (future) and, therefore, they become a measure of the adaptive
value of a game alternative for the system of players.

The game paths with higher assigned prices are the most desired by
the players' system as good adaptive solutions, therefore, these should
be more probable of being selected by the system of players, hence,
the probabilities of selection of each game path should numerically
coincide with the capitalized quantum Arrow-Debreu prices. This is
a numerical coincidence, not a conceptual one, since the capitalized
prices are expressed in units of payoff, while the probabilities are
pure numbers. Under the evolutionary interpretation, the present value
of the decision, therefore, leads to an expectation on the future
decisional moment discounted to the present.

We now provide for an example from quantum corporate finance.

\subsection{A Two-Company Game Example}

In order to illustrate the formalism introduced so far, let us assume
that two companies, labelled A and B, are evaluating a new investment
opportunity, having the choice of implementing the project with one
of two alternative technologies, labelled 0 and 1.

Both companies are expected to announce simultaneously, the new investment,
making public to the markets the technology chosen, and neither company
knows beforehand what the other will choose, so that each decides
without knowledge of the other's decision.

If both companies choose the same technology, then, company A expects
to obtain a project's return index (PRI) of 2, while it obtains a
PRI of 1.5 if the companies do not choose the same technology. Company
B, on the other hand, expects to obtain a PRI of: 2.5 if the companies
do not choose the same technology; 1.4 if both companies choose technology
0 and 2 if both companies choose technology 1.

There are four primitive chains for this game, with a freedom to choose
who to place first in the chain. In the present case, for simplicity,
we follow the labels and write these chains as $0\overset{A_{s}}{\longrightarrow}1\overset{B_{s}}{\longrightarrow}2$,
where $A_{s}$ means $A$ chooses technology $s$, with $s=0,1$,
the same holding for $B_{s}$. The game's choice paths' morphisms
are, then, expressed as $0\overset{B_{s}A_{r}}{\longrightarrow}2$,
with $B_{s}A_{r}=\left(\mathcal{H}_{G},\mathcal{P}_{G},\hat{P}_{B_{s}A_{r}}\right)$
and $\mathcal{P}_{G}=\left\{ \hat{P}_{B_{s}A_{r}}=\left|B_{s}A_{r}\right\rangle \left\langle B_{s}A_{r}\right|:r,s=0,1\right\} $.

Given the description of the game, the payoff operators are, in turn,
given by:
\begin{equation}
\hat{\pi}_{A}=2\left(\hat{P}_{B_{0}A_{0}}+\hat{P}_{B_{1}A_{1}}\right)+1.5\left(\hat{P}_{B_{1}A_{0}}+\hat{P}_{B_{0}A_{1}}\right)
\end{equation}
\begin{equation}
\hat{\pi}_{B}=1.4\hat{P}_{B_{0}A_{0}}+2\hat{P}_{B_{1}A_{1}}+2.5\left(\hat{P}_{B_{1}A_{0}}+\hat{P}_{B_{0}A_{1}}\right)
\end{equation}

Up to a price-conserving unitary transformation, the quantum Nash
equilibrium is, in this case, given by:

\begin{equation}
\left|Q\right\rangle =\sqrt{\frac{5}{32}}\left(\left|B_{0}A_{0}\right\rangle +\left|B_{1}A_{0}\right\rangle \right)+\sqrt{\frac{11}{32}}\left(\left|B_{0}A_{1}\right\rangle +\left|B_{1}A_{1}\right\rangle \right)
\end{equation}
which is separable in \emph{A} and \emph{B}'s strategies:
\begin{equation}
\left|Q\right\rangle =\left(\frac{1}{\sqrt{2}}\left|B_{0}\right\rangle +\frac{1}{\sqrt{2}}\left|B_{1}\right\rangle \right)\otimes\left(\sqrt{\frac{5}{16}}\left|A_{0}\right\rangle +\sqrt{\frac{11}{16}}\left|A_{1}\right\rangle \right)
\end{equation}
and the resulting pricing matrices reflect this tensor product structure:
\begin{equation}
\mathbf{D}_{A}=\frac{5}{16}\hat{P}_{A_{0}}+\frac{11}{16}\hat{P}_{A_{1}},\,\mathbf{D}_{B}=\frac{1}{2}\hat{P}_{A_{0}}+\frac{1}{2}\hat{P}_{A_{1}}
\end{equation}
Assuming that the IRP are calculated for the beginning of the game,
which forms part of the {}``year 0'' at which the project is evaluated,
we have $D=1$ and, thus, assuming the quantum Arrow-Debreu prices
to be expressed in monetary units: $PV_{A}=\$1.75$ while $PV_{B}=\$2.15625$.
Under the evolutionary framework, introduced in the previous subsection,
$A$ chooses technology $0$ with probability equal to $\frac{5}{16}$
and chooses technology $1$ with probability $\frac{11}{16}$, while
$B$ plays a fair coin game for $0$ and $1$, these probabilities
are conserved for any quantum Arrow-Debreu price conserving transformation
of $\left|Q\right\rangle $ and the choice of one player is probabilistically
independent from the choice of the other player.

If, on the other hand, \emph{company A} were to announce first its
decision, then, \emph{B} could always wait for \emph{A} and choose
its strategy so as to maximize its payoff, a resulting general solution,
for $D=1$, is an entangled ket:
\begin{equation}
\left|Q\right\rangle =\psi_{A}(0)\left|B_{1}A_{0}\right\rangle +\psi_{A}(1)\left|B_{0}A_{1}\right\rangle 
\end{equation}
such that $PV_{A}=\left\langle Q\left|\hat{\pi}_{A}\right|Q\right\rangle =1.5$
and $PV_{B}=\left\langle Q\left|\hat{\pi}_{B}\right|Q\right\rangle =2.5$,
which is the result of either the path $B_{1}A_{0}$ or the path $B_{0}A_{1}$,
being played by the companies, it does not matter what the value of
the quantum Arrow-Debreu prices are, the result is always the same
for each company.

\section{Quantum Games with Exchange Economies}

We now expand the formalism of the previous section to include exchange
economies, first a pure exchange economy comprised of a lottery that
is a single bet game on the result of an underlying quantum game of
strategy, and, second, a securities economy.

\subsection{Single bet game}

In the present section it is useful to introduce Greek letter indexes
ranging in $1,2,...N$, keeping the Latin lettered indexes for denoting
adaptive agents (both agents in the exchange economy as well as game
players), thus, one assumes to order the morphic fundaments $f_{n...21}$
and write, according to that order, $f_{\alpha}$ with the index $\alpha=1,2,...,N$.
Assuming this notation, in order to introduce an exchange economy
structure, we, first, expand the strategic game of \textbf{subsection
2.1} with a lottery, by adding an addicional $x_{n+1}$ object, such
that the morphism (5) is expanded to: 
\begin{equation}
x_{0}\overset{f_{n...21}}{\longrightarrow}x_{n}\overset{\omega}{\longrightarrow}x_{n+1}
\end{equation}
where $\omega$ is an index ranging as $\omega=1,2,...N$. This is
called the single bet game, since there is only one lottery type.
The ket (12) is, accordingly, replaced by: 
\begin{equation}
\left|Q\right\rangle =\frac{1}{\sqrt{D}}\sum_{\omega\alpha}\psi\left(f_{\alpha},\omega\right)\left|f_{\alpha}\omega\right\rangle 
\end{equation}
The price amplitudes $\psi\left(f_{\alpha},\omega\right)$ are now
pricing, simultaneously, the game path and the lottery.

Introducing the lottery economy operator $\hat{e}$, such that:

\begin{equation}
\hat{e}\left|Q\right\rangle =\left|Q\right\rangle 
\end{equation}

\begin{equation}
\hat{e}\left|f_{\alpha}\omega_{k}\right\rangle =\begin{cases}
\left|f_{\alpha}\omega\right\rangle , & \omega=\alpha\\
0, & \omega\neq\alpha
\end{cases}
\end{equation}
then, it follows that the lottery result is entangled with the game,
that is, in the quantum game equilibrium, we have the entangled kets
as solutions to the above eigenvalue equation:
\begin{equation}
\left|Q\right\rangle =\frac{1}{\sqrt{D}}\sum_{\omega}\psi\left(f_{\omega},\omega\right)\left|f_{\omega}\omega\right\rangle 
\end{equation}
thus, the quantum state for the game is no longer a pure state, but,
instead a statistical mixture density operator, resulting from tracing
out the lottery economy, such that, from the systemic position of
the game, we have the density operators for the game and the lottery:

\begin{equation}
\hat{\rho}_{Game}=Tr_{Lottery}\left(\left|Q\right\rangle \left\langle Q\right|\right)=\sum_{\omega}\frac{\left|\psi\left(f_{\omega},\omega\right)\right|^{2}}{D}\hat{P}_{f_{\omega}}
\end{equation}
\begin{equation}
\hat{\rho}_{Lottery}=Tr_{Game}\left(\left|Q\right\rangle \left\langle Q\right|\right)=\sum_{\omega}\frac{\left|\psi\left(f_{\omega},\omega\right)\right|^{2}}{D}\hat{P}_{\omega}
\end{equation}
in both cases, the capitalized quantum Arrow-Debreu prices $\frac{\left|\psi\left(f_{\omega},\omega\right)\right|^{2}}{D}$
are assumed from now on to be, always, the Nash Equilibrium prices
for the underlying quantum game of strategy.

Now, let us introduce a two-period pure exchange economy with a single
perishable good in both periods, such that the agents, in the economy,
choose a consumption at time $0$ (beginning of the game) and state
contingent claims on lottery-dependent consumption for the end of
the game, taken, for simplicity, to be time $1$. 

The exchange econonomy's agents and the players are, in this case,
assumed to be different entities, while the players are addressing
the game of strategy, the agents are addressing the lottery, through
a consumption allocation problem, such that, $c_{i0}$ is the \emph{i}-th
agent's consumption at time $0$, and $c_{i\omega}$ is the \emph{i}-th
agent's consumption at time $1$. A utility operator on consumption
for each agent is introduced such that, for $i=1,2,...,I$ ($I$ being
the number of agents in the exchange economy):

\begin{equation}
\hat{u}_{i}\left|f_{\alpha}\omega_{k}\right\rangle =u_{i}\left(c_{i0},c_{i\omega}\right)\left|f_{\alpha}\omega\right\rangle 
\end{equation}
For two different agents, their utility operators are assumed to commute
$\left[\hat{u}_{i},\hat{u}_{j}\right]=\delta_{ij}$ and the utility
functions $u_{i}$ are assumed to be increasing and strictly concave
functions of the consumption plan $\left(c_{i0},c_{i\omega}\right)$.
We now have a problem of allocation of state contingent consumption
among agents, where, without loss of generality, the single consumption
good is used as the numeraire for the exchange economy \cite{key-5}.

Taking $C_{0}$ to be the aggregate time-0 consumption available and
$C_{\omega}$ the aggregate consumption in the $\omega$-th lottery
result at time $1$, the feasibility conditions are given by \cite{key-5}:
\begin{equation}
C_{0}=\sum_{i=1}^{I}c_{i0},\: C_{\omega}=\sum_{i=1}^{I}c_{i\omega}
\end{equation}

Assuming that each agent knows that the lottery and the strategic
game are entangled and assigns a subjective probability to each alternative
$\omega$, then, each agent has a subjectively assigned density operator
for the lottery:

\begin{equation}
\hat{\rho}_{Lottery}^{i}=\sum_{\omega}bel_{i}\left(f_{\omega},\omega\right)\hat{P}_{\omega}
\end{equation}
where $bel_{i}\left(f_{\omega},\omega\right)$ is a subjective statistical
weight that the agent assigns to the lottery, without any further
information available upon the underlying game's payoff system. The
(subjective) expected utility is:

\begin{equation}
\left\langle \hat{u}_{i}\right\rangle _{Bel}=Tr\left(\hat{\rho}_{Lottery}^{i}\hat{u}_{i}\right)
\end{equation}

Then, the following optimization problem, then, ensues for the Pareto
optimal allocation:
\begin{equation}
\begin{array}{cc}
 & \max_{\left\{ \left(c_{i0},c_{i\omega}\right)_{i=1}^{I},\omega=1,2,...,N\right\} }\sum_{i=1}^{I}\lambda_{i}\left\langle \hat{u}_{i}\right\rangle _{Bel}\\
s.t. & \sum_{i=1}^{I}=C_{\omega},\:\omega=1,2,...,N\\
 & \sum_{i=1}^{I}c_{i0}=C_{0}
\end{array}
\end{equation}

Forming the Lagrangian for the optimization problem we have:
\begin{equation}
\begin{aligned}\max_{\left\{ \left(c_{i0},c_{i\omega}\right)_{i=1}^{I},\omega=1,2,...,N\right\} }L=\\
 & =\sum_{i=1}^{I}\lambda_{i}\left\langle \hat{u}_{i}\right\rangle _{Bel}+\phi_{0}\left[C_{0}-\sum_{i=1}^{I}c_{i0}\right]+\\
 & +\sum_{\omega=1}^{N}\phi_{\omega}\left[C_{\omega}-\sum_{i=1}^{I}c_{i\omega}\right]
\end{aligned}
\end{equation}

The first-order conditions lead to:
\begin{equation}
\begin{array}{cc}
\lambda_{i}\sum_{\omega=1}^{N}bel_{i}\left(f_{\omega},\omega\right)\frac{\partial u_{i}(c_{i0},c_{i\omega})}{\partial c_{i0}}=\phi_{0}, & i=1,2,...,I\end{array}
\end{equation}
\begin{equation}
\begin{array}{cc}
\lambda_{i}bel_{i}\left(f_{\omega},\omega\right)\frac{\partial u_{i}(c_{i0},c_{i\omega})}{\partial c_{i\omega}}=\phi_{\omega}, & \omega=1,2,...,N,\end{array}i=1,2,...,I
\end{equation}
\begin{equation}
\sum_{i=1}^{I}c_{i\omega}=C_{\omega},\:\forall\omega=1,2,...,N
\end{equation}
\begin{equation}
\sum_{i=1}^{I}c_{i0}=C_{0}
\end{equation}
Since the utility functions are assumed to be increasing and strictly
concave and the weights $\left\{ \lambda_{i}\right\} _{i=1}^{I}$
are assumed to be strictly positive, the first order conditions are
necessary and sufficient for a global maximum \cite{key-5}.

Replacing (42) in (41) for each agent, the marginal rates of substitution
between present consumption and future lottery-state contingent consumption
are equal across individuals:
\begin{equation}
\frac{bel_{i}\left(f_{\omega},\omega\right)\frac{\partial u_{i}(c_{i0},c_{i\omega})}{\partial c_{i\omega}}}{\sum_{\omega=1}^{N}bel_{i}\left(f_{\omega},\omega\right)\frac{\partial u_{i}(c_{i0},c_{i\omega})}{\partial c_{i0}}}=\frac{\phi_{\omega}}{\phi_{0}}
\end{equation}
for $\omega=1,2,...,N$ and $i=1,2,...,I$.

Following Huang and Litzenberger in \cite{key-5}, the above optimization
problem can be solved for the cases where, taking $\phi_{0}=1$, $\phi_{\omega}$
become the lottery-state contingent Arrow-Debreu prices. In the present
case, since the contingency does not come from a pure lottery, but,
instead from a lottery that is entangled with a finite quantum game
of strategy, we may let $\phi_{\omega}$ equal the quantum Arrow-Debreu
price for the game, that is: 
\begin{equation}
\phi_{\omega}=\left|\psi\left(f_{\omega},\omega\right)\right|^{2}
\end{equation}
 If we replace (46) in (45), and let $\phi_{0}=1$, we obtain:
\begin{equation}
\frac{bel_{i}\left(f_{\omega},\omega\right)\frac{\partial u_{i}(c_{i0},c_{i\omega})}{\partial c_{i\omega}}}{\sum_{\omega=1}^{N}bel_{i}\left(f_{\omega},\omega\right)\frac{\partial u_{i}(c_{i0},c_{i\omega})}{\partial c_{i0}}}=\left|\psi\left(f_{\omega},\omega\right)\right|^{2}
\end{equation}
going further, and assuming that each agent has enough information
to form a financially sustained rational equilibrium prediction about
the game, that is, provided each agent, in the transaction economy,
has enough information on the game of strategy to compute the quantum
Nash equilibrium for the game and, therefore, the probabilities for
the lottery, then:
\begin{equation}
bel_{i}\left(f_{\omega},\omega\right)=\frac{\left|\psi\left(f_{\omega},\omega\right)\right|^{2}}{D}
\end{equation}
Replacing this last result in (47), and given the positivity of the
quantum Arrow-Debreu prices, the following result holds:
\begin{equation}
\frac{\partial u_{i}(c_{i0},c_{i\omega})}{\partial c_{i\omega}}=D\sum_{\omega=1}^{N}\frac{\left|\psi\left(f_{\omega},\omega\right)\right|^{2}}{D}\cdot\frac{\partial u_{i}(c_{i0},c_{i\omega})}{\partial c_{i0}}
\end{equation}
establishing the relation between the marginal utilities and the quantum
game of strategy solution.

\subsection{Securities market game}

Assuming the single-lottery framework of equations (30) to (32), we
can introduce a transaction market for $m$ securities with payoff
operators:
\begin{equation}
\hat{x}_{j}\left|f_{\alpha}\omega_{k}\right\rangle =\begin{cases}
x_{j}\left(\omega\right)\left|f_{\alpha}\omega\right\rangle , & \omega=\alpha\\
0, & \omega\neq\alpha
\end{cases}
\end{equation}
for $j=1,2...,m$ . Each operator is financially consistent with the
lottery economy, in the sense that $\left[\hat{x}_{j},\hat{e}\right]=0$
having the same entangled eigenstates described by equation (32).
Letting $s_{ij}$ and $S_{j}$ denote, respectively, the number of
securities held by the agent \emph{i}, in equilibrium, and the price
for security \emph{j} at time 0, the relevant optimization problem
for each agent's portfolio, given the previous subsection's framework,
can be formulated as follows \cite{key-5}:
\begin{equation}
\begin{array}{cc}
 & \max_{\left\{ c_{i0},s_{ij};j=1,2,...,m\right\} }\left\langle \hat{u}_{i}\right\rangle _{Bel}\\
s.t. & c_{i0}+\sum_{j}s_{ij}S_{j}=e_{i0}+\sum_{j=1}^{m}w_{ij}S_{j}
\end{array}
\end{equation}
where $e_{i0}$ is an endowment of time 0 consumption and $w_{ij}$
is a time 0 endowment of shares of security $j$. Given the previously
assumed properties of the utility functions, the necessary and sufficient
conditions for the \emph{i}-th agent's portfolio problem are given
by:
\begin{equation}
\sum_{\omega=1}^{m}\frac{bel_{i}\left(f_{\omega},\omega\right)\frac{\partial u_{i}(c_{i0},c_{i\omega})}{\partial c_{i\omega}}}{\sum_{\omega=1}^{m}bel_{i}\left(f_{\omega},\omega\right)\frac{\partial u_{i}(c_{i0},c_{i\omega})}{\partial c_{i0}}}x_{j}(\omega)=S_{j}
\end{equation}
with $c_{i\omega}:=\sum_{j=1}^{m}s_{ij}x_{j}\left(\omega\right)$.

To derive a Pareto optimal allocation that reflects the underlying
game of strategy structure, it is first, important to notice that,
given the entangled ket of equation (32), one may express a financial
valuation from the density operator for the lottery such that, for
the game's quantum Nash equilibrium, we have:
\begin{equation}
Tr\left(\hat{\rho}_{Lottery}\hat{x}_{j}\right)=\sum_{\omega=1}^{N}x_{j}(\omega)\frac{\left|\psi\left(f_{\omega},\omega\right)\right|^{2}}{D}
\end{equation}
this is the expected payoff for the \emph{j}-th security, then, following
\textbf{section 2.}'s framework, one may calculate a present value
as one would for any other player's position, and write:

\begin{equation}
PV_{j}=D\times Tr\left(\hat{\rho}_{Economy}\hat{x}_{j}\right)=D\sum_{\omega=1}^{N}x_{j}(\omega)\frac{\left|\psi\left(f_{\omega},\omega\right)\right|^{2}}{D}
\end{equation}
it is then (financially) natural for the securities market agents
to value the security in terms of this game payoff, such that the
condition holds:
\begin{equation}
S_{j}=PV_{j}
\end{equation}
which is financially sound since each state-contingent payoff is weighed
by the quantum game price of that state, the state being, in this
case, the result of the lottery which is, in turn, contingent upon
the game of strategy's path.

Now, if the number of securities is equal to the number of players
in the underlying game of strategy ($m=n$), and if each security
is a financially transactionable exposure to the corresponding player's
position in the game of strategy (\emph{securitization of game positions}),
in this case, one may consider that the security payoff is proportional
to the player's payoff, the constant of proportionality $\theta$
giving the exposure (for instance, the proportion of equity is an
example of an exposure to a company's value, the company being the
player, the equity holders being the financial agents and the shares
being the securities), formally, we may write $\hat{x}_{i}:=\theta\hat{\pi}_{i}$.

Thus, by playing for the quantum Nash equilibrium the players are
simultaneously maximizing the value of the corresponding security
asset, and therefore, the portfolio value for any financial agent
is maximized by the Nash optimizing adaptive behavior of the players
involved in the underlying game of strategy.

Denoting $\left|\psi\left(f_{\omega},\omega\right)\right|^{2}$ by
$\phi_{\omega}$ and letting, as before, $bel_{i}\left(f_{\omega},\omega\right)=\phi_{\omega}$,
by replacing (54) in (52), we obtain the optimum:

\begin{equation}
\sum_{\omega=1}^{N}\phi_{\omega}\frac{\frac{\partial u_{i}(c_{i0},c_{i\omega})}{\partial c_{i\omega}}}{\sum_{\omega=1}^{N}\phi_{\omega}\frac{\partial u_{i}(c_{i0},c_{i\omega})}{\partial c_{i0}}}x_{j}(\omega)=D\sum_{\omega=1}^{N}\frac{1}{D}\phi_{\omega}x_{j}(\omega)
\end{equation}
which can be rewritten as:

\begin{equation}
\sum_{\omega=1}^{N}\phi_{\omega}x_{j}(\omega)\left(\frac{\frac{\partial u_{i}(c_{i0},c_{i\omega})}{\partial c_{i\omega}}}{D\sum_{\omega=1}^{N}\phi_{\omega}\frac{\partial u_{i}(c_{i0},c_{i\omega})}{\partial c_{i0}}}-\frac{1}{D}\right)=0
\end{equation}
One solution, in particular, for (56) results from taking the argument
within the brakets to be zero, $\frac{\frac{\partial u_{i}(c_{i0},c_{i\omega})}{\partial c_{i\omega}}}{D\sum_{\omega=1}^{m}\phi_{\omega}\frac{\partial u_{i}(c_{i0},c_{i\omega})}{\partial c_{i0}}}-\frac{1}{D}=0$,
which leads to the inter-temporal Pareto optimal consumption plan
condition: 

\begin{equation}
\frac{\partial u_{i}(c_{i0},c_{i\omega})}{\partial c_{i\omega}}=D\sum_{\omega=1}^{m}\frac{1}{D}\phi_{\omega}\frac{\partial u_{i}(c_{i0},c_{i\omega})}{\partial c_{i0}}
\end{equation}
Therefore, even in incomplete securities' markets, as long as there
is a security exposure for each game player's position that is proportional
to the payoffs on that position, and assuming that the quantum Nash
equilibrium is played, one obtains a consumption Pareto optimal plan
for securities prices reflecting the quantum Nash equilibrium of the
game.

The optimal consumption plan depends upon two behaviors: a quantum
Nash optimizing behavior on the part of the players and an equilibrium
valuation behavior on the part of the securities' market agents, such
that the securities' market agents reflect, in each securities' price,
the financial present value of the exposure at the corresponding quantum
Nash equilibrium. Therefore, the adaptiveness of the financial market
system is inextricably interweaved with the adaptiveness of the business
game playing system.

\section{Conclusion}

Financial decisions, regarding the inter-temporal allocation of wealth
and of financial resources, does not take place within a single pure
lottery game structure, rather, while financial agents manage asset
portfolios, trying to optimize their intertemporal securities' management,
the value drivers for these assets come from a complex adaptive dynamics
in which companies play games of strategy with their stakeholders,
adapting to opportunities and responding to threats, expanding their
strengths and addressing their weaknesses, managing their business
risk. Business strategic decisions affect directly the probabilities
associated with different financial scenarios and, therefore, directly
affect asset allocation and valuation issues.

This state of affairs (business and financial) was the main point
of concern for the present article, dealing with the effects upon
the traditional financial optimization problems, when one abandons
the assumption of a pure lottery game towards a lottery whose results,
and therefore probability profiles, are entangled with a game of strategy
where players enact a quantum Nash equilibrium.

In order to operationalize this problem it became necessary to address
how one might obtain an Arrow-Debreu price structure from the game
of strategy itself and, at the same time, entangle that structure
with a lottery game, such that all of the financial lottery state-contingent
prices coincide with the quantum game equilibrium solution Arrow-Debreu
prices.

Such a financial approach to game theory, along with the quantum game
equilibrium theorem, derived in \textbf{section 2.}, thus, allowed
for the quantum formalism to be applied towards a greater effectiveness
in the integration between business game theory and financial decision
theory, showing how Pareto optimal allocations in a securities market
may reflect the game of strategy's quantum Nash equilibrium playing,
a result that is independent from any completeness assumption that
would restrict the number of linearly independent securities to equal
the number of lottery states.

\end{document}